\documentclass{llncs}

\usepackage{makeidx}
\usepackage{graphicx}
\usepackage{amsmath}
\usepackage{amssymb}
\usepackage{url}

\newtheorem{Definition}{Definition}

\usepackage{algorithm}
\usepackage{algorithmic}

\begin{document}

\title{LGM: Mining Frequent Subgraphs from Linear Graphs}
\author{Yasuo Tabei\inst{1} \and Daisuke Okanohara\inst{2} \and Shuichi Hirose\inst{3} \and Koji Tsuda\inst{1,3}}
\authorrunning{Yasuo Tabei et al.}
\tocauthor{Yasuo Tabei, Daisuke Okanohara, Shuichi Hirose, Koji Tsuda}
\institute{ERATO Minato Project, Japan Science and Technology Agency, Sapporo, Japan \\
\and
Preferred Infrastructure, Inc, Tokyo, Japan \\
\and
Computational Biology Research Center, National Institute of Advanced Industrial Science and Technology (AIST), Tokyo, Japan\\
\email{yasuo.tabei@gmail.com}, \email{hillbig@preferred.jp}, \email{shuichi.hirose@nagase.co.jp}, \email{koji.tsuda@aist.go.jp}
}
\maketitle

%\pagenumbering{arabic}
%\setcounter{page}{1}%Leave this line commented out.

\begin{abstract} %\small\baselineskip=9pt
A linear graph is a graph whose vertices are totally ordered. 
Biological and linguistic sequences with interactions among symbols are
naturally represented as linear graphs. Examples include
protein contact maps, RNA secondary structures and predicate-argument
structures. 
Our algorithm, linear graph miner (LGM), leverages the vertex order for
efficient enumeration of frequent subgraphs. 
Based on the reverse search principle, the pattern space 
is systematically traversed without expensive duplication checking.
Disconnected subgraph patterns are particularly important in linear
graphs due to their sequential nature. 
Unlike conventional graph mining algorithms detecting connected patterns
only, LGM can detect disconnected patterns as well.
The utility and efficiency of LGM are demonstrated in experiments 
on protein contact maps.
% and natural language sentences.
\end{abstract}

\section{Introduction}
Frequent subgraph mining is an active research area with successful
applications in, e.g., 
chemoinformatics~\cite{Saigo08}, software science~\cite{Eic08}, and computer vision~\cite{Now07a}.
The task is to enumerate the complete set of frequently appearing subgraphs in
a graph database. 
Early algorithms include AGM~\cite{Inokuchi00}, FSG~\cite{Kuramochi01}
and gSpan~\cite{Yan02}.
Since then, researchers paid considerable efforts to improve the
efficiency, for example, by mining closed patterns only~\cite{Yan03}, or 
by early pruning that sacrifices the completeness (e.g., leap search~\cite{yan08}).
However, graph mining algorithms are still too slow for large graph databases (see e.g.,\cite{wale06}).
The scalability of graph mining algorithms is much worse than those for
more restricted classes such as trees~\cite{abe02} and sequences~\cite{pei01}.
It is due to the fact that, for trees and sequences, it is
possible to design a pattern extension rule that does not create
duplicate patterns (e.g., rightmost extension)~\cite{abe02}.
For general graphs, there are multiple ways to generate the same
subgraph pattern, and it is necessary to detect duplicate patterns and
prune the search tree whenever duplication is detected.
In gSpan~\cite{Yan02}, a graph pattern is represented as a DFS code, 
and the duplication check is implemented via minimality checking of the code. 
It is a very clever mechanism, because one does not need to track back
the patterns generated so far.
Nevertheless, the complexity of duplication checking is 
exponential to the pattern size~\cite{Yan02}.
It harms efficiency substantially, especially when mining large patterns.

A linear graph is a graph whose vertices 
are totally ordered~\cite{Dav06,Fer07} (Figure~\ref{fig:lg}). 
For example, protein contact maps, RNA secondary structures, alternative splicing patterns 
in molecular biology and predicate-argument structures~\cite{Miyao08} in natural
languages can be represented as linear graphs.
Amino acid residues of a protein have natural ordering 
from N- to C-terminus, and English words in a sentence are ordered as well.
Davydov and Batzoglou~\cite{Dav06} addressed the problem of aligning several
linear graphs for RNA sequences, assessed the computational
complexity, and proposed an approximate algorithm. 
Fertin et al. assessed the complexity of finding a maximum common
pattern in a set of linear graphs~\cite{Fer07}.
In this paper, we develop a novel algorithm, linear graph miner (LGM), 
for enumerating frequently appearing subgraphs in 
a large number of linear graphs. 
The advantage of employing linear graphs is that we can derive a
pattern extension rule that does not cause duplication, 
which makes LGM much more efficient than conventional graph mining algorithms.

\begin{figure}[t]
\begin{center}
\begin{tabular}{cc}
\includegraphics[width=0.4\textwidth]{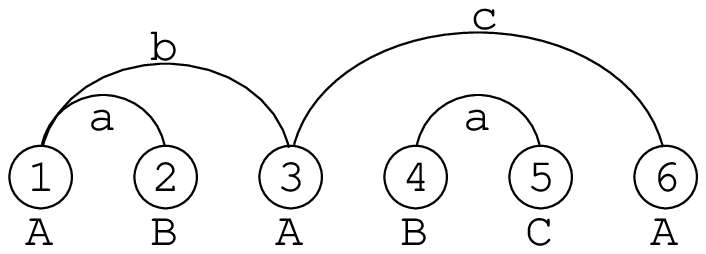}
\end{tabular}
\vspace{-0.5cm}
\end{center}
\caption{An example of linear graph}
\label{fig:lg}
\end{figure}

We design the extension rule based on the reverse search 
principle~\cite{Avis96}. Perhaps confusingly, 'reverse search' does not 
refer to a particular search method, but a guideline for designing
enumeration algorithms. 
A pattern extension rule specifies how to generate children from a
parent in the search space. 
In reverse search, one specifies a rule that generates a parent uniquely
from a child (i.e., {\em reduction map}).
The pattern extension rule is obtained by 'reversing' the reduction map:
When generating children from a parent, all possible candidates are
prepared and those mapping back to the parent by the reduction map are selected.
An advantage of reverse search is that, given a reduction map, 
the completeness of the resulting pattern extension rule 
can easily be proved~\cite{Avis96}. 
In data mining, LCM, one of the fastest closed itemset miner, 
was designed using reverse search~\cite{Uno05}. 
It is applied in the design of a dense module enumeration 
algorithm~\cite{Geo09}
and a geometric graph mining algorithm recently~\cite{Now08}.
In computational geometry and related fields, 
there are many successful applications.\footnote{See a list of
  applications at \url{http://cgm.cs.mcgill.ca/~avis/doc/rs/applications/index.html}}
LGM's reduction map is very simple: remove the largest edge 
in terms of edge ordering.
Fortunately, it is not necessary to take the 
``candidate preparation and selection'' approach in LGM. 
We can directly reverse the reduction map to an explicit extension rule here.

Linear graphs can be perceived as the fusion of graphs and sequences. 
Sequence mining algorithms such as Prefixspan~\cite{pei01} can usually detect 
gaped sequence patterns. In applications like motif discovery in
protein contact maps~\cite{Glyakina07}, it is essential to allow ``gaps'' 
in linear graph patterns. More precisely, {\em disconnected} graph patterns 
should be allowed for such applications. 
Since conventional graph mining algorithms can detect only
connected graph patterns, their application to contact maps is
difficult. 
In this paper, we aim to detect connected and disconnected
patterns with a unified framework.

In experiments, we used a protein 3D-structure dataset 
from molecular biology.
We compared LGM with gSpan in efficiency, and found that 
LGM is more efficient than gSpan.
It is surprising to us, because LGM detects a much larger number of patterns
including disconnected ones.  
To compare the two methods on the same basis, we added supplementary
edges to help gSpan to detect a part of disconnected patterns.
Then, the efficiency difference became even more significant.

\section{Preliminaries}
Let us first define linear graphs and associated concepts.
\begin{Definition} [Linear graph]
{\rm
Denote by $\Sigma^V$ and $\Sigma^E$ the set of vertex and edge
labels, respectively.
A labeled and undirected {\em linear graph} 
$g=(V, E, L^{V},L^{E})$ 
consists of an ordered vertex set $V \subset \mathbb{N}$, 
an edge set $E \subseteq V \times V$, a vertex labeling 
$L^{V}: V \rightarrow \Sigma^{V}$ and an edge labeling 
$L^{E}: E \rightarrow \Sigma^{E}$. 
Let the size of the linear graph $|g|$ be the number of its edges.
Let ${\mathcal G}$ denote the set of all possible linear graphs 
and let $\theta \in {\mathcal G}$ denote the empty graph.
}
\end{Definition}

\noindent
The difference from ordinary graphs is that the vertices are defined 
as a subset of natural numbers, introducing the total order.
Notice that we do not impose connectedness here.
%In this paper, the graph size $|g|$ refers to 
%the number of edges $E_g$.
%Let ${\mathcal G}$ denote the set of all possible linear graphs and 
%let $\phi \in {\mathcal G}$ denote the empty graph. 
%We write an edge between vertexes $i$ and $j$, $i < j$, as a pair $(i,j)$.
%\begin{figure}[t]
%\begin{center}
%\begin{tabular}{cc}
%\includegraphics[width=0.4\textwidth]{fig/projected_graph.eps}
%\end{tabular}
%\end{center}
%\caption{Example of a projected graph. Disconnected pairs of vertexes and their edge are projected into a projected graph. 
%The disconnected pairs of which distances are within $\delta$ are projected.}
%\label{fig:lg}
%\end{figure}
The order of edges is defined as follows:
\begin{Definition}[Total order among edges]\label{def:order}
{\rm
$\forall e_1=(i, j), e_2=(k, l) \in E_{g}$, $e_1 <_e e_2$ 
if and only if i) $i < k$ or ii) $i = k, j < l$. }
\end{Definition}
Namely, one first compares the indices of the left nodes. If they are
identical, the right nodes are compared. The subgraph relationship
between two linear graphs is defined as follows. 
\begin{Definition}[Subgraph]
{\rm
Given two linear graphs $g_1=(V_{1}, E_{1}, L^{V_{1}}, L^{E_{1}})$, 
$g_2=(V_{2}, E_{2}, L^{V_{2}}, L^{E_{2}})$,  
$g_1$ is a subgraph of $g_2$, $g_1 \subseteq g_2$,  
if and only if there exists 
an injective mapping $m: V_{1} \rightarrow V_{2}$  such that 
\begin{enumerate}
  \item $\forall i \in V_{1}: L^{V_1}(i) = L^{V_2}(m(i)),$ vertex labels are identical, 
  \item $\forall (i,j) \in E_{1}: (m(i), m(j)) \in E_{2}, L^{E_{1}}(i,j) = L^{E_{2}}(m(i),m(j))$, all edges of $g_1$ exist in $g_2$, and
  \item $\forall (i,j) \in E_{1}: i < j \rightarrow m(i) < m(j)$, the order of vertices is conserved.
\end{enumerate}
}
\end{Definition}
The difference from the ordinary subgraph relation is that the vertex
order is conserved. Finally, frequent subgraph mining is defined as follows.
\begin{Definition}[Frequent linear subgraph mining] \label{Frequentmining}
{\rm
%\paragraph{Problem1 Frequent linear subgraph mining}
For a set of linear graphs $G=\{g_1, \cdots, g_{|G|}\}, 
g_i \in {\mathcal G}$, 
a minimum support threshold $\sigma > 0$ and 
a maximum pattern size $s > 0$, 
find all $g \in {\mathcal G}$ such that $g$ is frequent enough in $G$, i.e.,
\begin{eqnarray}
  |\{i=1,...,|G|:g \subseteq g_i\}| \ge \sigma, |g| \leq s \nonumber
\end{eqnarray}
}
\end{Definition}

\begin{figure*}[t]
\begin{center}
\begin{tabular}{cc}
\includegraphics[width=0.4\textwidth]{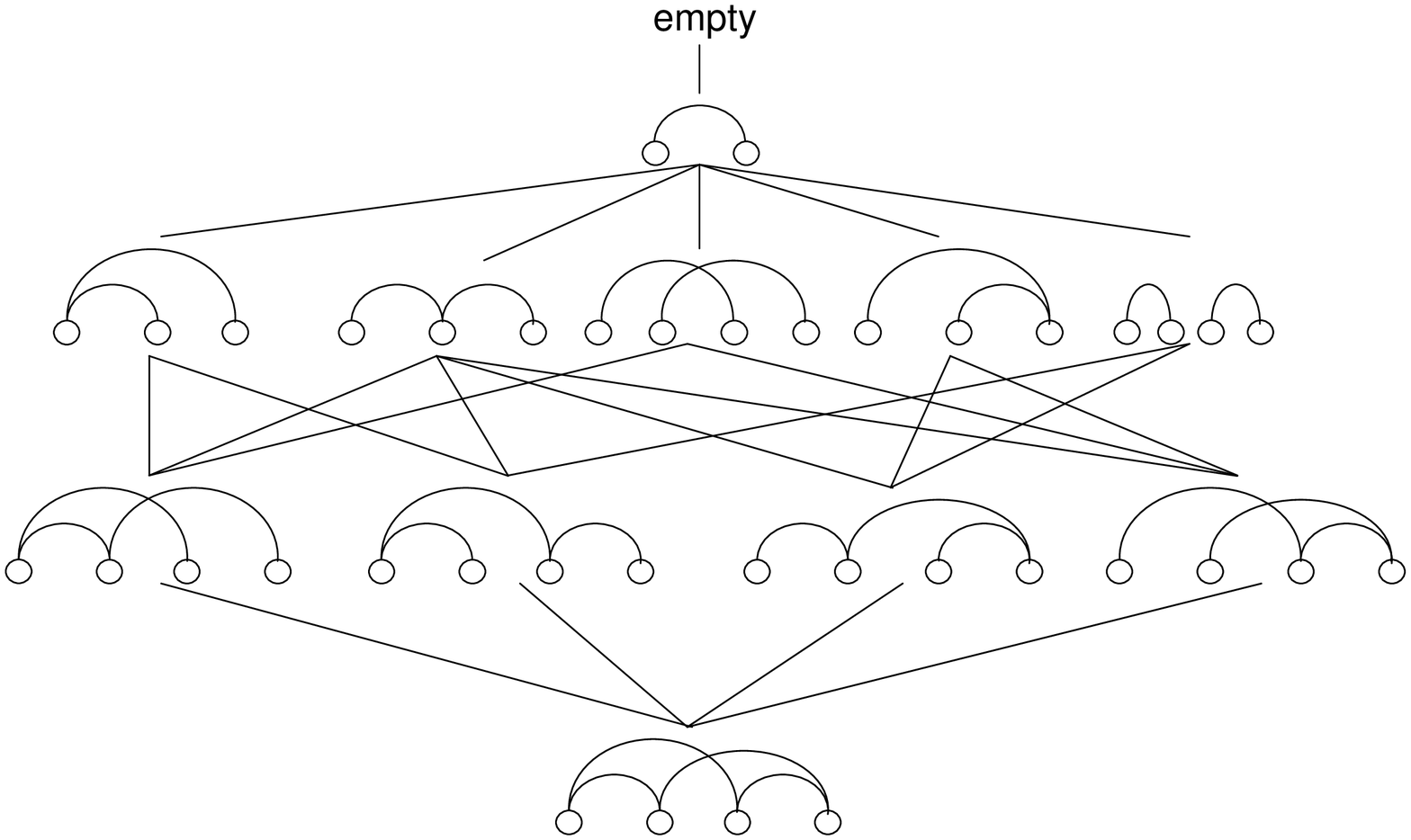} &
\includegraphics[width=0.4\textwidth]{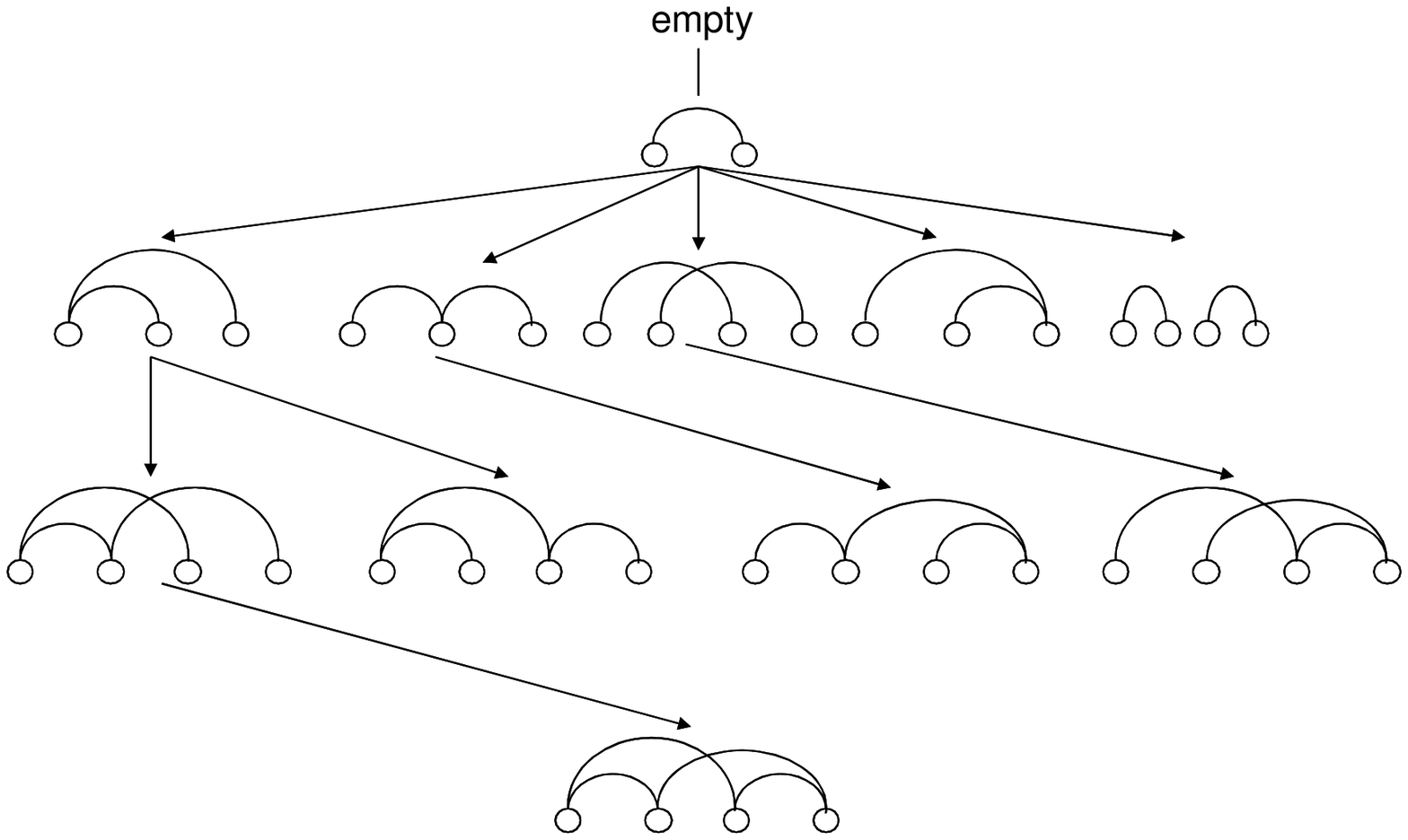}
%\vspace{-0.5cm}
\end{tabular}
\end{center}
\caption{(Left) Graph-shaped search space. 
(Right) Search tree induced by the reduction map}
\label{fig:graph_shaped_search_space}
\end{figure*}

\begin{figure*}[t]
\centering
\includegraphics[width=0.95\textwidth]{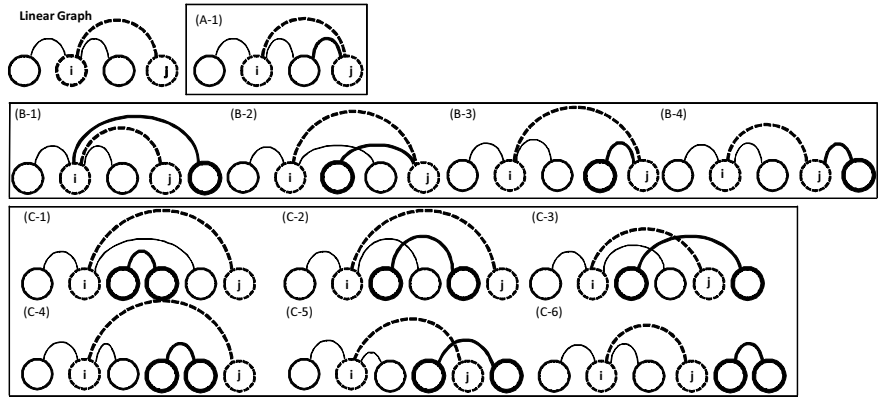}
%\vspace{-0.5cm}
\caption{Example of children patterns. There are three types of extension with respect to the number of nodes: 
(A) no-node-addition, (B) one-node-addition, (C) two-nodes-addition.}
\label{fig:childgen}
\end{figure*}

\section{Enumeration of Linear Subgraphs}
Before addressing the frequent pattern mining problem,
let us design an algorithm for enumerating all subgraphs of a linear
graph. For simplicity, we do not consider vertex and edge labels in this
section, but inclusion of the labels is straightforward.

\subsection{Reduction Map}
Suppose we would like to enumerate all subgraphs in 
a linear graph shown in the bottom of 
Figure~\ref{fig:graph_shaped_search_space}, left.
All linear subgraphs form a natural graph-shaped search space, 
where one can traverse upwards or downwards by deleting or adding an edge 
(Figure \ref{fig:graph_shaped_search_space}, left). 
For enumeration, however, 
one has to choose edges in the search graph to form a search tree (Figure 2, right).
Once a search tree is defined, the enumeration can be done 
either by depth-first or breadth-first traversal.
To this aim, we specify a {\em reduction map} 
$f: {\mathcal G} \rightarrow {\mathcal G}$ 
which transforms a child to its parent uniquely. 
The mapping is chosen such that when it is applied 
repeatedly, we eventually reduce it to an element of the solution set
${\mathcal S} \subset {\mathcal G}$. Formally, we write 
$\forall x \in {\mathcal G}: \exists k \leq 0: f^k(x) \in {\mathcal S}$.
In our case, the reduction map is defined as 
removing the ``largest'' edge from the child graph. 
The largest edge is defined via the total order 
introduced in Definition~\ref{def:order}. 
By evaluating the mapping repeatedly the graph is shrunk 
to the empty graph. Thus, here we have ${\mathcal S=\{\theta\}}$.

By applying $f(g)$ for all possible $g \in {\mathcal G}$, 
we can induce a search tree with 
$\theta \in {\mathcal G}$ being the root node, 
shown in Figure~\ref{fig:graph_shaped_search_space}, right.
A question is if we can always define a unique search tree for 
any linear graph. 
The reverse search theorem~\cite{Avis96} says that 
the proposition is true iff any node in the graph-shaped search space
converges to the root node (i.e., empty graph) by applying the map 
a finite number of times. For our reduction map, 
it is true, because each possible 
linear graph $g \in {\mathcal G}$ 
is reduced to the empty graph by successively applying 
$f$ to $g$.

A characteristic point of reverse search is that the search tree 
is implicitly defined by the reduction map.
In actual traversal, the search tree is created on demand: 
when a traverser is at a node with graph $g$ and would like to move down, a set of children nodes are generated by extending $g$. 
More precisely, one enumerate all linear graphs 
by inverting the reduction mapping such that the tree is explored 
from the root node towards the leaves. 

The inverse mapping $f^{-1}:{\mathcal G} \rightarrow {\mathcal G^*}$ 
generates for a given linear graph $g \in {\mathcal G}$ a 
set of extended graphs $X=\{g^\prime \; | \; f(g^\prime) = g\}$.

There are three types of extension patterns according to the number of added nodes in the reduction mapping:
(A) no-node-addition, (B) one-node-addition, (C) two-nodes-addition.
Let us define the largest edge of $g$ as $({\bold i}, {\bold j}), {\bold i} < {\bold j}$. 
Then, the enumeration of case A is done by adding an edge which is larger than $({\bold i}, {\bold j})$.
For case B, a node is inserted to the position after ${\bold i}$, 
and this node is connected to every other node.
If the new edge is smaller than $({\bold i}, {\bold j})$, this extension is canceled.
For case C, two nodes are inserted to the position after ${\bold i}$. 
In that case, the added two nodes must be connected by a new edge.
All patterns of valid extensions are shown in Figure~\ref{fig:childgen}.
This example does not include node labels, 
but for actual applications, node labels need to be enumerated as well. 
\section{Frequent Pattern Mining} 
In frequent pattern mining, we employ the same search tree described above, 
but the occurrence of a pattern in all linear graphs are tracked 
in an {\em occurrence list} $L_{G}(g)$~\cite{Yan02}, defined as follows:
\begin{eqnarray}
  L_G(g) &=& \{(i,m):g_i \in G, g \subseteq g_i~\mbox{with node correspondence}~m \}. \nonumber
\end{eqnarray}
When a pattern $g$ is extended, its occurrence list $L_G(g)$ is updated as well. 
Based on the occurrence list, the support of each pattern $g$, i.e., 
the number of linear graphs which contains the pattern, is calculated.
Whenever the support is smaller than the threshold $s$, 
the search tree is pruned at this node.
This pruning is possible, because of the anti-monotonicity of the support, 
namely the support of a graph is never larger than that of its subgraph. 
Algorithm \ref{LGM} describes the recursive algorithm for frequent mining. 
In line 13, each pattern $g$ is extended to larger graphs 
$g^\prime \in f^{-1}(g)$ 
by inverse reduction mapping $f^{-1}$. 
The possible extensions $f^{-1}(g)$ for each pattern $g$ are found 
using the location list $L_G(g)$.
The function {\sc Mine} is recursively called for each extended pattern $g' \in f^{-1}(g)$ in line 15. The graph pruning happens in lines 7, 
if the support for the pattern $g$ is smaller than 
the minimum support threshold $\sigma$ or 
in line 11 if the pattern size $|g|$ is equal to the maximum pattern size $s$.

\begin{algorithm}[t]
  \caption{Linear Graph Miner (LGM)}
  \label{LGM}
  \begin{algorithmic}[0]
    \STATE {\bf Input}:
    \STATE ~~ A set of linear graphs: $G = \{g_1,...,g_{|G|}\}$
    \STATE ~~ Minimum support: $\sigma \geq 0$
    \STATE ~~ Maximum pattern size: $s \geq 0$
  \end{algorithmic}
  \begin{algorithmic}[1]
    \STATE {\bf function} {\sc LGM}($G$, $\sigma$, $s$)~~~~~~~$\triangleright$ {\small the main function}
    \STATE ~~ {\sc Mine}$(G, \phi, \sigma, s)$
    \STATE {\bf end~function}
    \STATE {\bf function} {\sc Mine}$(G, g, \sigma, s)$
    \STATE ~~ $sup \leftarrow {\rm support}(L_G(g))$
    \STATE ~~ {\bf if} $sup < \sigma$ {\bf then}~~~~~~~~~$\triangleright${\small check support condition}
    \STATE ~~~~ {\bf return}
    \STATE ~~ {\bf end if}
    \STATE ~~{\bf Report} occurrence of subgraph $g$
    \STATE ~~ {\bf if} $|g| = s$ {\bf then}~~~~~~~~~~~~~~~~~~~~$\triangleright${\small check pattern size}
    \STATE ~~~~ {\bf return}
    \STATE ~~ {\bf end if}
    \STATE ~~ scan $G$ once by using $L_G(g)$, find all extensions $f^{-1}(g)$
    \STATE ~~ {\bf for} $g^\prime \in f^{-1}(g)$
    \STATE ~~~~ {\sc Mine}$(G,g^\prime,\sigma,s)$\\~~~~~~~~~~~~~~~~~~~~~$\triangleright${\small call {\sc Mine} for every extended pattern $g'$}
    \STATE ~~ {\bf end for}
    \STATE {\bf end~function}
  \end{algorithmic}

\end{algorithm}

\section{Complexity Analysis} 
The computational time of frequent pattern mining depends 
on the minimum support and maximum pattern size thresholds~\cite{Yan02}. 
Also, it depends on the ``density'' of the database: 
If all graphs are almost identical (i.e., a dense database), 
the mining would take a prohibitive amount of time. 
So, conventional worst case analysis is not amenable to mining algorithms. 
Instead, the {\em delay}, interval time between two consecutive solutions, 
is often used to describe the complexity. 
General graph mining algorithms including gSpan are exponential delay 
algorithms, i.e., the delay is exponential to the size of patterns~\cite{Yan02}. The delay of our algorithm is only polynomial, 
because no duplication checks are necessary thanks to the vertex order. 
\begin{theorem}[Polynomial delay]
For $N$ linear graphs $G$, a minimum support $\sigma > 0$, and a maximum pattern size $s > 0$, the time between two successive calls to 
${\bf Report}$ in line 9 is bounded by a polynomial of the size of input data.
\end{theorem}

\begin{proof} Let $M:=\max_i|V_{g_i}|$, $F:=\max_i|E_{g_i}|$. 
The number of matching locations in the linear graphs $G$ can decrease in case $g$ is enlarged, 
because the only largest edge is added. Considering the number of 
variations, it is easy to see that the location list always satisfies
$|L_G(g)| \leq M^2N$. Therefore, the mapping $f^{-1}(g)$ can be produced in $O(M^2N)$ time, 
because the procedure searches for the location list in line 13. 

The time complexity between two successive calls to ${\sc Report}$ can now 
be bounded by considering two cases after ${\sc Report}$ has been called once.
\begin{itemize}
  \item Case 1. There is an extension $g'$ fulfilling the minimum support condition,or the size of $g'$ is $s$. Then ${\sc Report}$ is called within $O(M^2N)$ time.
  \item Case 2. There is no extension $g'$ fulfilling the minimum support condition.Then, no recursion happens and ${\sc Mine}$ returns in $O(M^2N)$ time 
to its parent node in the search tree. The maximum number of times this can 
happen successively is bounded by the depth of the reverse search tree, 
which is bounded by $O(F)$, because each level in the search tree adds one edge. 
Therefore, in $O(M^2NF)$ time the algorithm either calls ${\sc Report}$ again or finishes. 
\end{itemize}
Thus, the total time between two successive calls to ${\sc Report}$ is bounded by $O(M^2NF)$.
\end{proof}

\begin{figure}[t]
\begin{center}
  \includegraphics[width=0.6\textwidth]{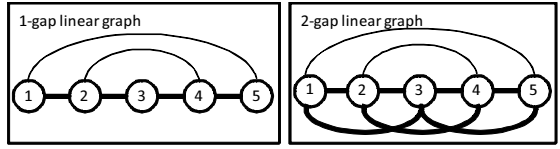} 
\caption{Example of gap linear graph. 1-gap linear graph~(left) and 2-gap linear graph~(right) are represented, respectively. Edges corresponding to gaps are represented in bold line.}
\label{fig:gappedgraph}
\end{center}
\end{figure}

\section{Experiments}
We performed a motif extraction experiment from protein 3D
structures. Frequent and characteristic patterns are often called
``motifs'' in molecular biology, and we adopt that terminology here. 
All experiments were performed on a Linux machine with an AMD Opteron processor (2 GHz and 4GB RAM).

\subsection{Motif extraction from protein 3D structures}
We adopted the Glyankina et al's dataset~\cite{Glyakina07} 
which consists of pairs of homologous proteins: one is 
derived from a thermophilic organism and the other is from 
a mesophilic organism. 
This dataset was made for understanding structural properties of proteins which are responsible 
for the higher thermostability of proteins from thermophilic organisms compared to those from mesophilic organisms.
In constructing a linear graph from a 3D structure,
each amino acid is represented as a vertex.
Vertex labels are chosen from $\{1,\ldots,6\}$, 
which represents the following six classes:  aliphatic \{AVLIMC\}, 
aromatic \{FWYH\}, polar \{STNQ\}, positive \{KR\}, negative \{DE\}, 
special (reflecting their special conformation properties)
\{GP\}~\cite{Mirny99}.
An edge is drawn between the pair of amino acid residues whose distance is within 5 angstrom. 
No edge labels are assigned. 
In total, $754$ graphs were made.
Average number of vertices and edges are $371$ and $498$, respectively, 
and the number of labels is $6$.
To detect the motifs characterizing the difference between two
organisms, we take the following two-step approach.
First, we employ LGM to find frequent patterns from all proteins of both 
organisms. 
In this setting, we did not use (c-6) patterns in Figure 3. 
Finally, the patterns significantly associated with organism difference are
selected via statistical tests.

\begin{figure}[t]
   \begin{minipage}{.4\textwidth}
     \includegraphics[width=5cm]{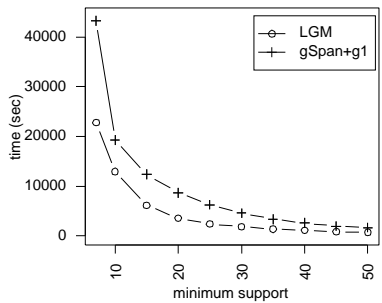}
   \end{minipage}
   \begin{minipage}{.1\textwidth}
    \hfill
   \end{minipage}
   \begin{minipage}{.4\textwidth}
\vspace*{-1cm}
      \caption{Execution time for the protein data. 
      The line labeled by gSpan+g1 is execution time for gSpan on the 1-gap linear graph dataset.
      gSpan does not work on the 2-gap linear graph dataset even if the minimum support threshold is 50.}
\label{fig:protein_time}
   \end{minipage}

\end{figure}

%\begin{figure}[t]
%\begin{center}
%  \includegraphics[width=0.4\textwidth]{fig/lgm_vs_gspan_time.eps} 
%\caption{Execution time for the protein data. The line labeled by gSpan+g1 is execution time for gSpan on the 1-gap linear graph dataset. gSpan does not work on the 2-gap linear graph dataset even if the minimum support threshold is 50.}
%\label{fig:protein_time}
%\end{center}
%\end{figure}

\begin{figure}[t]
\begin{center}
\begin{tabular}{cc}
\includegraphics[width=0.7\textwidth]{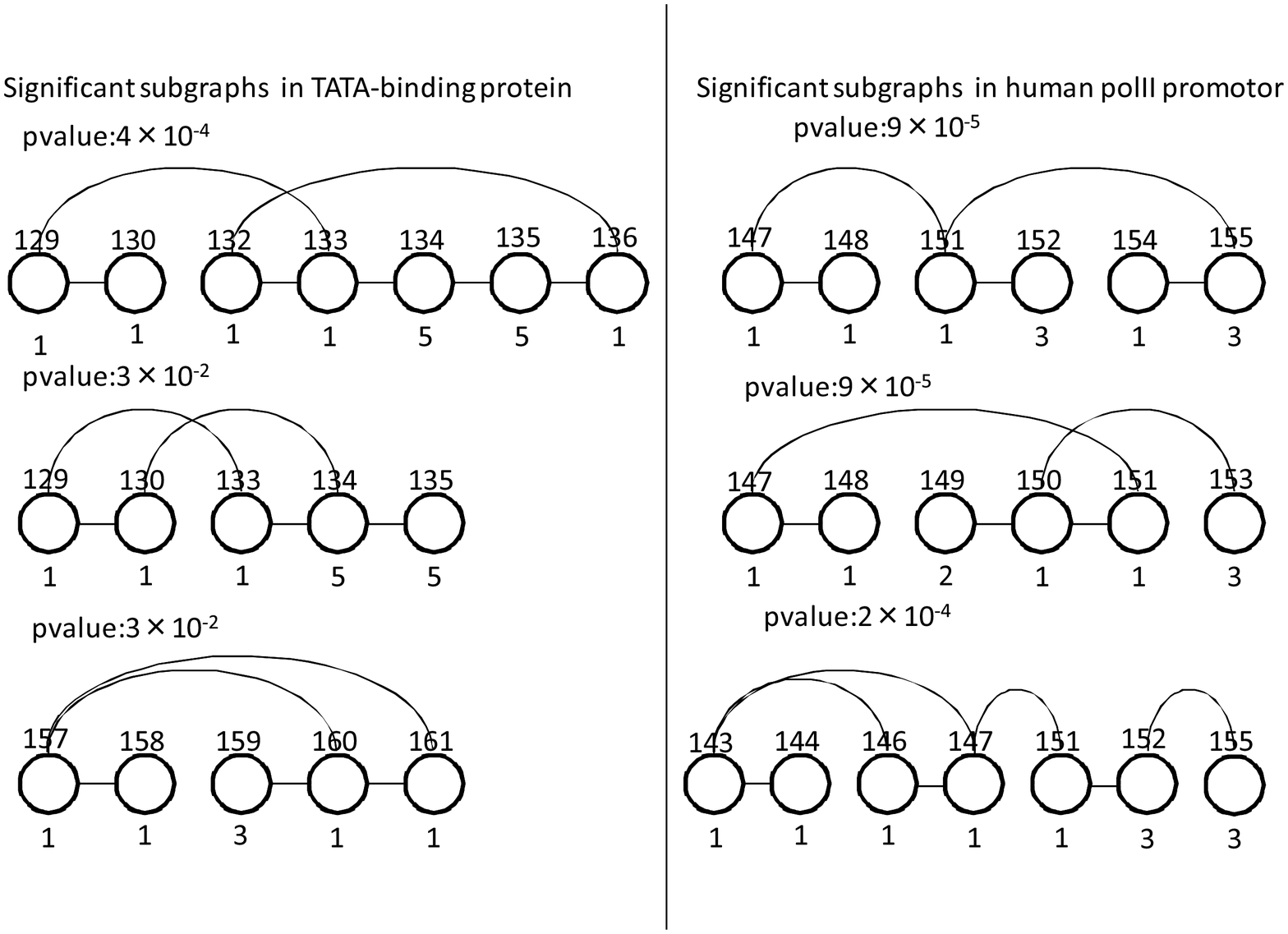} &
\end{tabular}
\vspace{-0.5cm}
\end{center}
\caption{Significant subgraphs detected by LGM. The p-value calculated by fisher exact test is attached to each linear graph. The node labels 1, 2, 3, 4 and 5, represent aliphatic, aromatic, polar, positive and negative proteins, respectively.}
\label{fig:sigpro}
\end{figure}

\begin{figure}
\begin{center}
\begin{tabular}{cc}
\includegraphics[width=0.4\textwidth]{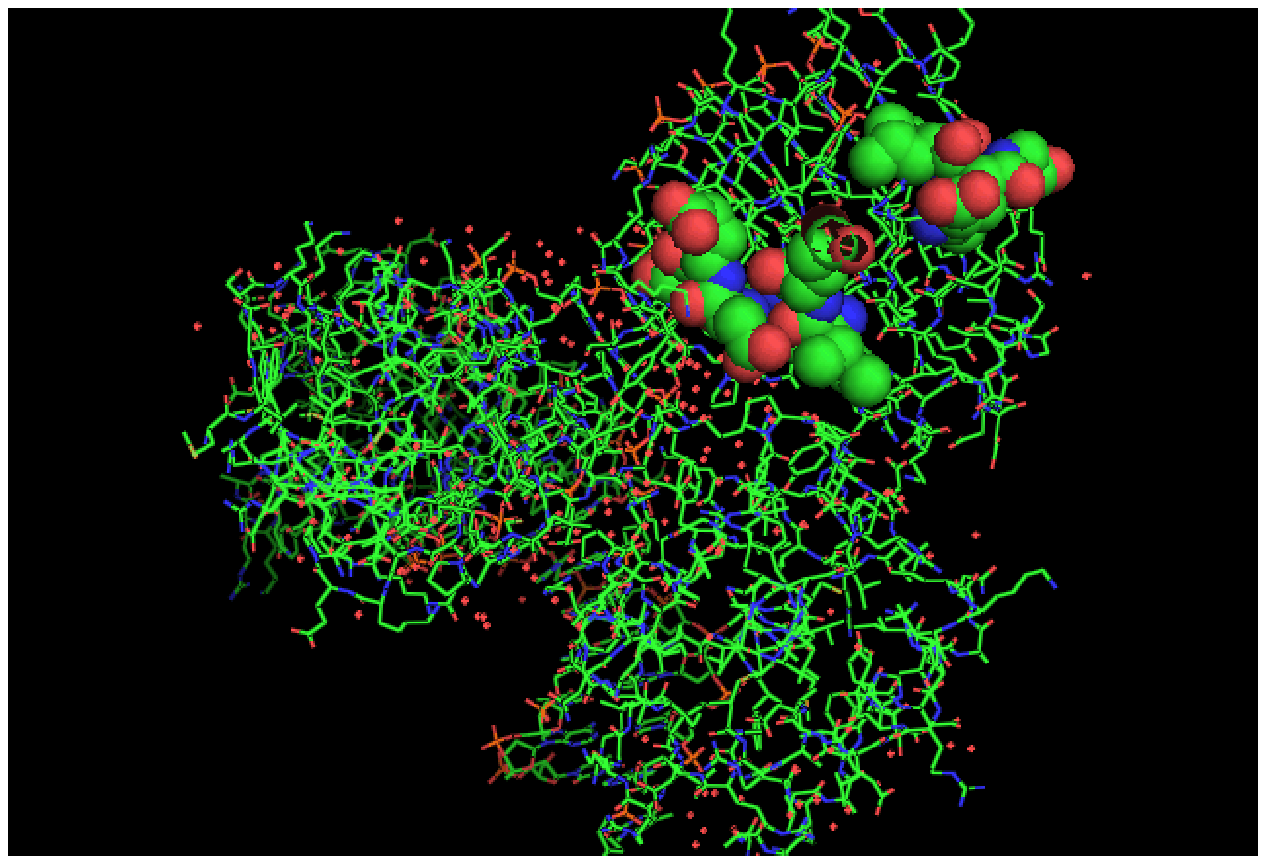} &
\includegraphics[width=0.4\textwidth]{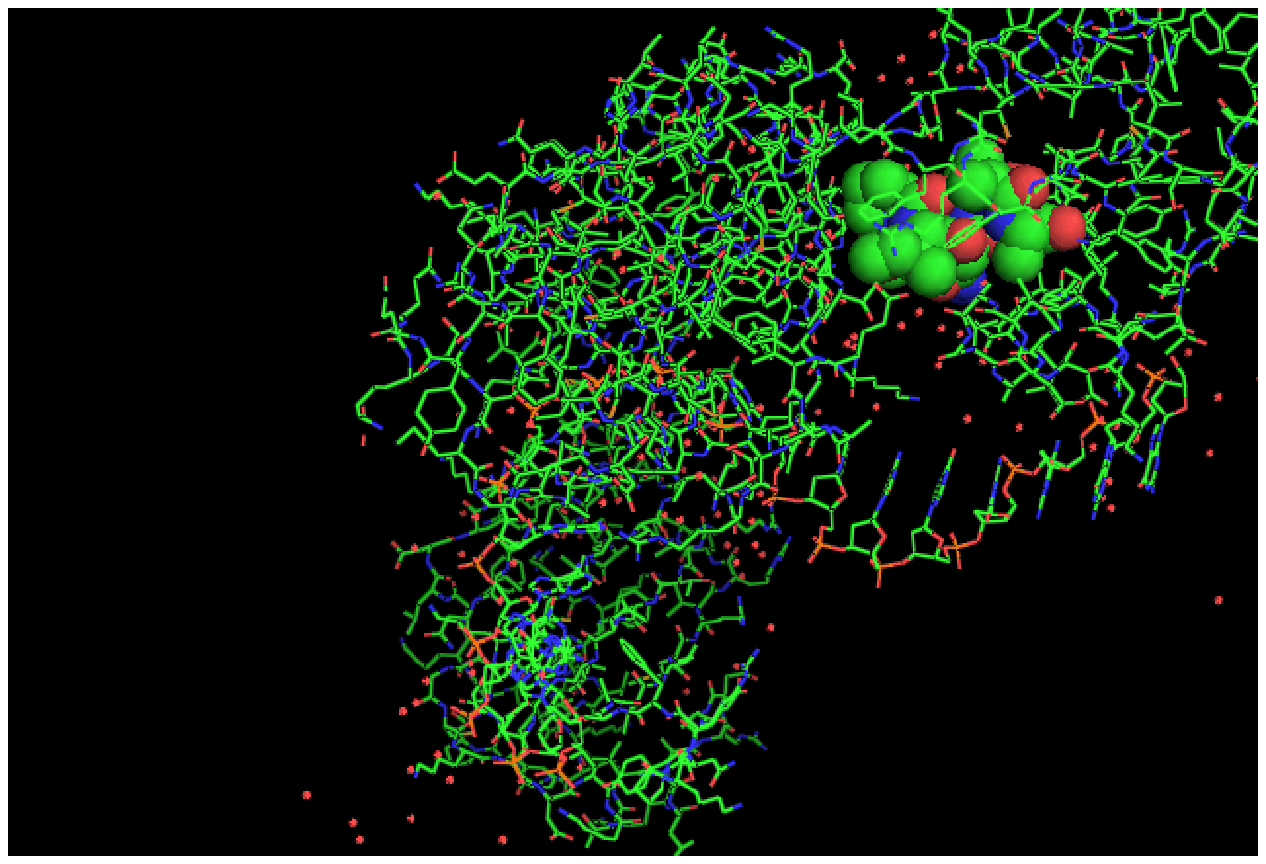} 
\end{tabular}
\vspace{-0.5cm}
\end{center}
\caption{3D-structures of TATA-binding protein(left) and human pol II promotor protein (right). The spheres represent the amino acid residues corresponding to vertices forming subgraphs in figure \ref{fig:sigpro}. }
\label{fig:protein}
\end{figure}

We assess the execution time of our algorithm in comparison with gSpan. 
The linear graphs from 3D-structure proteins are 
not always connected graphs and 
the gSpan can not be applied to such disconnected graphs. 
Hence, we made two kinds of gaped linear graph: 1-gap linear graph and 2-gap linear graph. 
1-gap linear graph is a linear graph whose contiguous vertices in a protein sequence are connected by an edge;
2-gap linear graph is a 1-gap linear graph whose two vertices skipping one in a protein sequence are 
connected by an edge (Figure~\ref{fig:gappedgraph}).
We run gSpan on two datasets: one consists of 1-gap linear graphs 
and the other consists of 2-gap linear graphs.
We run LGM on the original linear graphs.
We set the maximum execution time to 12 hours for both programs.
Figure~\ref{fig:protein_time} shows the execution time 
by changing minimum support thresholds. 
gSpan does not work on the 2-gap linear graph dataset even if the minimum support threshold is 50.
Our algorithm is faster than gSpan on the 1-gap linear graph dataset, and 
its execution time is reasonable. 

%\begin{table}[t]
%\begin{center}
%\begin{tabular}{cccccc} 
%\hline
% & PAS+LGM & PAS+gSpan & dep+FREQT & n-gram+FREQT & bow\\
%\hline
%sentiment & 100745 (2,10) & 59763(2,10) & 108704 (2,inf) & 53112 (2,inf) & 21425 \\
%subjectivity & 2898222 (2,5) & 1108032(2,10) & 122974 (2,10) & 60247 (2,inf) & 23926 \\
%\hline
%\end{tabular}
%\end{center}
%\caption{The number of enumerated patterns for sentiment and subjectivity data. 
%The parenthetic numbers represent a minimum support threshold (left) and a maximum% pattern size (right) for 
%a mining algorithm.}
%\label{table:num_enum_pat}
%\end{table}

Then, we assess a motif extraction ability of our algorithm. 
To choose significant subgraphs from the enumerated subgraphs, 
we use Fisher's exact test. 
In this case, a significant subgraph should distinguish thermophilic proteins
from mesophilic proteins. 
Thus, for each frequent subgraph, we count the number of proteins 
containing this subgraph in the thermophilic and mesophilic proteins; and 
generate a $2 \times 2$ contingency table, which includes the number of thermophilic organisms that contain subgraph $g$ $n_{TP}$, the number of thermophilic organisms that does not contain a subgraph $g$ $n_{FP}$, the number of mesophilic organisms that does not contain a subraph $g$ $n_{FN}$ and the number of mesophilic organisms that contain a subgraph $g$ $n_{TN}$. 
The probability representing the independence in the contingency table is calculated as follows:
\[
  Pr = \frac{\begin{pmatrix} n_g \\ n_{TP} \end{pmatrix}\begin{pmatrix}n_{g'} \\ n_{FN} \end{pmatrix}}{\begin{pmatrix} n \\ n_p \end{pmatrix}} 
    = \frac{n_g! n_{g'}!n_P!n_N!}{n!n_{TP}!n_{FP}!n_{FN}!n_{TN}!}, 
\]
where $n_{P}$ is the number of thermophilic proteins;
$n_{N}$ the number of mesophilic proteins; 
$n_{g} $ the number of proteins with a subgraph $g$; 
$n_{g'}$ the number of proteins without a subgraph $g'$. 
The p-value of the two-sided Fisher's exact test on a table can be computed by the 
sum of all probabilities of tables that are more extreme than this table. 

We ranked the frequent subgraphs according to the p-values, 
and obtained 103 subgraphs whose p-values are no more than 0.001.
% shold write the number of significant classes
Here, we focused on a pair of proteins, 
TATA-binding protein and human polII promotor protein, 
where TATA-binding protein is derived from a thermophilic organism 
and human polII promotor is from a mesophilic organism. 
The reason we chose these two proteins is that they 
include a large number of statistically significant motifs 
which are mutually exclusive between two organisms. 
These two proteins share the same function as DNA-binding protein, 
but their thermostabilities are different. 
Figure \ref{fig:sigpro} shows the top-3 subgraphs in significance. 
%We could obtain characteristic subgraphs with respect 
%to thermostability by using our algorithm. 
Figure \ref{fig:protein} shows 3D-structure proteins, 
TATA-binding protein (left) and human polII promotor protein(right), 
and the amino acid residues forming top3-subgraphs are represented by spheres.

\section{Conclusion}
We proposed an efficient frequent subgraph mining algorithm from linear graphs.
A key point is that vertices in a linear graph are totally ordered. 
We designed a fast enumeration algorithm from 
linear graphs based on this property. 
For an efficient enumeration without duplication, 
we define a search tree based on reverse search techniques.
Different from gSpan, our algorithm enumerates frequent subgraphs including 
disconnected ones by traversing this search tree. 
Many kinds of data, 
such as protein 3D-structures and alternative splicing forms, 
which can be represented as linear graphs, 
include disconnected subgraphs as important patterns. 
The computational time of our algorithm is polynomial-delay. 

We performed a motif extraction experiment of a protein 3D-structure dataset 
in molecular biology.
% and a sentiment/subjectivity dataset in natural language processing. 
In the experiment, our algorithm could extract important subgraphs as 
frequent patterns. 
By comparing our algorithm to gSpan with respect to execution time, 
we have shown our algorithm is fast enough for the real world datasets. 

Data which can be represented as linear graphs occur in many fields, for instance bioinformatics and 
natural language processing. Our mining algorithm from linear graphs provide a new way to analyze such data.

\section*{Acknowledgements}
This work is partly supported by research fellowship from JSPS for young scientists,
MEXT Kakenhi 21680025 and the FIRST program.
We would like to thank 
M.~Gromiha for providing the protein 3D-structure dataset,
T.~Uno and H.~Kashima for fruitful discussions.

\bibliographystyle{plain}
\bibliography{biblio}

\end{document}